\documentclass[letterpaper,twocolumn,10pt]{article}
\usepackage{usenix-2020-09}

\usepackage{lastpage}

\usepackage[caption=false,font=footnotesize,labelfont=sf,textfont=sf]{subfig}

\usepackage{amsmath}
\usepackage{amssymb}
\usepackage{amsthm}
\usepackage{mathtools}

\DeclarePairedDelimiter\floor{\lfloor}{\rfloor}
\newtheorem{definition}{Definition}
\newtheorem{lemma}{Lemma}
\newtheorem{theorem}{Theorem}
\newtheorem{corollary}{Corollary}

\usepackage[noabbrev,capitalize]{cleveref}
\usepackage{url}

\usepackage{algpseudocodex}

\usepackage{graphicx}
\graphicspath{{./figures/}}
\DeclareGraphicsExtensions{.pdf,.jpeg,.png}

\usepackage{tikz}
\usetikzlibrary{shapes,arrows}
\tikzset{algpxIndentLine/.style={draw=lightgray}}
\definecolor{lightgray}{HTML}{D1DBE4}
\definecolor{green}{HTML}{B9C89F}
\definecolor{red}{HTML}{F19F9E}
\definecolor{yellow}{HTML}{F6DA71}
\definecolor{blue}{HTML}{ACE7EF}
\definecolor{darkyellow}{HTML}{FFA600}
\definecolor{darkblue}{HTML}{057DCD}
\definecolor{darkpurple}{HTML}{58508D}
\definecolor{darkred}{HTML}{FF6361}

\usepackage{booktabs}
\usepackage{tabularx}
\usepackage{todonotes}
\usepackage{xspace}

\newcommand{\sysname}{\textsc{Tyche}\xspace}
\newcommand{\sysnamezkp}{\textsc{Tyche-ZKP}\xspace}
\newcommand{\sysnamewul}{\textsc{Tyche-Coop}\xspace}

\newcommand{\x}{\mathbf{x}}

\begin{document}

\title{\sysname: Collateral-Free Coalition-Resistant Multiparty Lotteries\\with Arbitrary Payouts}

% \author{{\rm Anonymous Author(s)}}

\author{{\rm Quentin Kniep}\\
ETH Zurich\\
Switzerland\\
qkniep@ethz.ch
\and
{\rm Roger Wattenhofer}\\
ETH Zurich\\
Switzerland\\
wattenhofer@ethz.ch}

\maketitle

\begin{abstract}
We propose \sysname,
a family of protocols for performing practically (as well as asymptotically) efficient multiparty lotteries,
resistant against aborts and majority coalitions.
Our protocols are based on a commit-and-reveal approach,
requiring only a collision-resistant hash function.

All our protocols use a blockchain as a public bulletin board and for buy-in collection and payout settlement.
Importantly though, they do not rely on it or any other third party for providing randomness.
Also, participants are not required to post any collateral beyond their buy-in.
Any honest participant can eventually settle the lottery,
and dishonest behavior never reduces the winning probability of any honest participant.

Further, we adapt all three protocols into anonymous lotteries,
where (under certain conditions) the winner is unlinkable to any particular participant.
We show that our protocols are secure, fair, and some preserve the participants' privacy.

Finally, we evaluate the performance of our protocols, particularly in terms of transaction fees,
by implementing them on the Sui blockchain.
There we see that per user transaction fees are reasonably low and our protocols could potentially support millions of participants.
\end{abstract}

%\quentin{This doc has \pageref{lastpagebeforerefs}/\pageref{LastPage} out of 13/20 pages.}

\section{Introduction}
\label{sec:introduction}

Lotteries are interesting applications in the blockchain setting.
They are interesting in their own right, since they show the expressiveness of blockchains.
Also, they may be used as a mechanism for leader election in proof-of-stake systems,
or as part of higher-level applications that require fair (weighted) random distribution of some resource.
Lotteries can also be used with negative utilities to solve the leader aversion problem,
i.e., assign some duty randomly to one or some of the participants.

Various simplifying assumptions are often made to realize lotteries in the blockchain setting.
This is also true, more generally, for any applications requiring unbiasable randomness.
Many such applications running on blockchains use committee signatures from the consensus protocol or a third-party committee to provide the randomness~\cite{bitcoin_source_of_randomness,ethereum_lotteries_1,ethereum_lotteries_2}.
In many cases, this allows the block producer to bias the randomness output by deciding whether to censor a given block or which signatures to include~\cite{malleability_of_blockchain_entropy}.
Winning a large-scale lottery could be worth orders of magnitude more than the protocol rewards of any specific block.
Thus, it might be feasible for well-funded lottery participants to bribe block producers or randomness committee members.
To at least partially prevent this, heavy cryptographic primitives such as distributed verifiable random functions~\cite{vrf,distributed_vrf}.
Given any (practically) unbiasable source of randomness, however,
a lottery can be easily realized by generating a random permutation of the participants.
Other protocols do rely on randomness from participants but require a majority of honest participants~\cite{large-scale_p2p_lottery}.

Also, some earlier implementations of decentralized lotteries~\cite{fair_bitcoin_protocols,bitcoin_secure_mpc,constant_deposit_bitcoin_lotteries}
and other applications using distributed randomness~\cite{decentralized_poker,fastkitten}
rely on the fact that users pay a collateral in addition to their buy-in.
Usually the collateral is at least proportional to the number of participants.
Behavior deviating from the protocol can then be heavily punished.

\subsubsection*{Our Contributions}

In this work, we propose \sysname, a family of multiparty lotteries supporting arbitrary payout functions.
They provide general constructions from any correct monotonic shuffling network, which we also define.
These protocols require a public bulletin board, which can be instantiated by a smart contract on a blockchain.
Importantly, they make minimal assumptions, using the blockchain only for settlement and broadcast,
not as a source of randomness of any kind.
Instead, the outcome relies only on randomness chosen by the participants.
Like the most closely related protocols by Miller and Bentov~\cite{zero_collateral_lotteries} and Ballweg et al.~\cite{purelottery},
the \sysname family of protocols does not require any additional collateral to be posted.
Misbehaving participants are simply considered to be resigning and are removed from the lottery accordingly.

\section{Related Work}
\label{sec:related-work}

In general, Cleve~\cite{limits_on_secure_coin_flips} showed that any $r$-round coin flipping protocol allows the adversary to bias the output by at least $\Omega(1/r)$.
Similar bounds have been explored specific to the multiparty case~\cite{fair_coin_flipping}.
However, this impossibility can be circumvented by allowing biasing in only one direction.

% It has been shown~\cite{rational_secret_sharing}\ldots
% Further,~\cite{rss_revisited,efficient_rss}\ldots

Notions of fairness in multiparty coin flipping have been formally defined by Chung et al.~\cite{game_theoretic_fairness}.
In addition to fairness,
some lotteries~\cite{banfel,fplotto,anonymous_pos_lottery} also provide certain privacy guarantees to their participants.

In the blockchain setting, randomness can be derived from a committee~\cite{bitcoin_source_of_randomness,ethereum_lotteries_1,ethereum_lotteries_2}.
The committee serves as a third-party source of randomness.
It can for example be realized by relying on the validators of the underlying blockchain.
Many applications in practice are of this kind.
Other protocols use user-supplied randomness, like~\cite{fastkitten,fairlotto,prob_smart_contracts,large-scale_p2p_lottery,decentralized_poker,bitcoin_secure_mpc,fair_bitcoin_protocols}.
However, these rely on at least one of
(i) users posting collateral that is seized in case of misbehavior,
or (ii) no collusion between participants and a semi-honest lottery agency or between a (large) majority of participants.

For perfect game-theoretic fairness of such protocols,
Chung et al.~\cite{fair_leader_election} showed a communication lower bound of $\Omega(\log n)$ sequential rounds.
On the other hand, they also showed a way of closely approximating fairness in $\mathcal{O}(\log \log n)$ rounds.

We will define shuffling network,
which are closely modeled after the notion of sorting networks from distributed computing theory~\cite{sorting_networks}.
Efficient constructions (especially considering depth) have previously been shown~\cite{random_sorting_networks,optimal_sorting_networks,improved_logn_depth_sorting_networks}.

First, Delmonilo et al.~\cite{two_player_lotteries} proposed a zero-collateral lottery for two players (i.e. coin flip).
Previously, Miller and Bentov~\cite{zero_collateral_lotteries} proposed a protocol for realizing lotteries in the blockchain setting.
Their protocol does not require the participants to post any collateral.
However, it still has quite a few restrictions:
(i) number of participants needs to be a power of two,
(ii) only a single price is paid out,
(iii) all participants have equal probability of winning.
Ballweg et al.~\cite{purelottery} generalized this approach to non power-of-two numbers of participants.
Ballweg~\cite{purelottery_thesis} also partially alleviated the restrictions regarding winning probabilities and the multi-winner setting.

\section{Preliminaries}
\label{sec:preliminaries}

In this section, we introduce some fundamental mathematical notation and cryptographic definitions and constructions used throughout this work.
\\
\\
\textbf{Notation}.
For any $n \in \mathbb{N}$ we denote by $[n]$ the set $\{1,2,\ldots,n\}$.
If not otherwise specified, $\log$ refers to the base-2 logarithm.
We universally use $\lambda$ to denote a security parameter.
By $\textsf{poly}(x_1, \ldots, x_k)$ we denote values which are chosen as a polynomial function of the $k$ variables $x_i$.

\subsection{Cryptographic Definitions}

A fundamental cryptographic primitive used in \sysname protocols are cryptographic hash functions.

\begin{definition}[Cryptographic Hash Function]
    Given a security parameter $\lambda$.
    A cryptographic hash function is a function $h: \{0, 1\}^{*} \rightarrow \{0, 1\}^l$ with $l \ge \lambda$,
    which  achieves the following:
    \begin{itemize}
        \item \textbf{Preimage Resistance:} Given a hash $y \in \{0, 1\}^l$.
            It is infeasible to find an $x$ such that $h(x) = y$.
        \item \textbf{Second-preimage Resistance:} Given a hash $y \coloneqq h(x)$.
            It is infeasible to find any $x' \neq x$ with $h(x') = h(x) = y$.
    \end{itemize}
    We call it a strong (or collision resistant) cryptographic hash function if it also achieves:
    \begin{itemize}
        \item \textbf{Collision Resistance:} It is infeasible to find any $x_1, x_2$,
            such that $x_1 \neq x_2$ and $h(x_1) = h(x_2)$.
    \end{itemize}
    Collision resistance additionally requires $l \ge 2\lambda$.
\end{definition}

The central building block of \sysname protocols is a commitment scheme,
which in turn --- in the specific instantiation we present here --- is built on a cryptographic hash function.

\begin{definition}[Commitment Scheme]
    A commitment scheme provides two algorithms: $\mathsf{Commit}(x) \rightarrow c$ and $\mathsf{Verify}(c, x)$.
    We define the following security properties:
    \begin{itemize}
        \item \textbf{Perfectly Hiding:} It is information theoretically impossible for even a computationally unbounded adversary to determine input $x$ from commitment $c$.
        \item \textbf{Perfectly Binding:} For any commitment $c$ over an input $x$, there exists no $x' \neq x$ in the input space for which $\mathsf{Verify}(c, x')$ holds.
        \item \textbf{Computationally Hiding:} Under some cryptographic hardness assumption it is infeasible to determine input $x$ from commitment $c$.
        \item \textbf{Computationally Binding:} Under some cryptographic hardness assumption it is infeasible to open a commitment $c$ initially calculated over input $x$ to an input $x' \neq x$.
    \end{itemize}
\end{definition}

A commitment scheme may be (a) perfectly hiding and computationally binding or (b) perfectly binding and computationally hiding.
However, it can by definition not be both perfectly hiding and perfectly binding.

Using cryptographic hash functions in the random oracle model it is trivial to construct a commitment scheme,
which is at least computationally hiding and computationally binding.
Even in general, one-way functions imply a corresponding construction of a commitment scheme~\cite{commitment_from_prng,prng_from_owf}.

Achieving the strongest notions of privacy (see \cref{sec:unlinkable-lotteries})
requires the use of heavier cryptographic primitives, specifically zero-knowledge proofs.

\begin{definition}[Zero-Knowledge Proof]
\label{def:zkp}
    A non-interactive zero-knowledge proof scheme provides two algorithms:
    \begin{itemize}
        \item $\mathsf{Prove}(x, w) \rightarrow \pi$:
            Generate a proof $\pi$, given a witness $w$ making the statement $x$ evaluate to true.
        \item $\mathsf{Verify}(x, \pi)$:
            Verify that proof $\pi$ is valid for statement $x$.
    \end{itemize}
    It has to satisfy the following security properties:
    \begin{itemize}
        \item \textbf{Completeness:} For any valid combination of statement $x$ and witness $w$, $\pi \coloneqq \mathsf{Prove}(x, w)$ passes verification.
        \item \textbf{Soundness:} It is infeasible to generate a proof $\pi$ for a statement $x$ without knowing a valid witness $w$.
        \item \textbf{Zero-Knowledge:} $\pi$ reveals nothing about the secret input $sec$,
            apart from that it satisfies $circuit$.
    \end{itemize}
    We call it a non-malleable zero-knowledge proof if the scheme achieves the following stronger security property.
    \begin{itemize}
        \item \textbf{Non-Malleability:} It is infeasible for an adversary to transform a valid proof $\pi$ for a statement $x$
            into a valid proof $\pi'$ for a different statement $x'$.
    \end{itemize}
\end{definition}

Examples of practical schemes for zero-knowledge proofs are those by Groth~\cite{groth_zkp} and Gabizon et al.~\cite{plonk}.
However, these systems rely on hardness assumptions such as the discrete-logarithm problem and are thus not post-quantum secure.
ZK-STARKs~\cite{zk_stark}, which only rely on collision-resistant hash functions and the random oracle model,
might provide a viable alternative in a post-quantum future.
A practical example of ZK-STARK would be Zilch~\cite{zilch_zkp}.

\section{Collateral-Free Lotteries}
\label{sec:collateral-free-lotteries}

In this section, we define the model assumptions and properties we want from our lottery protocols.
Further, we reiterate constructions from the most closely related previous work, which operate in the same setting.
Later sections will build upon the ideas presented here.

\subsection{Model}
\label{sec:model}

Assume the existence of a public bulletin board.
Main functionality of the bulletin board is the ability to publish and read authenticated messages.
The messages are also assigned low-fidelity monotonic timestamps.
This allows protocols to roughly split time artificially into distinct rounds based on these timestamps.
Specifically, each message, after being posted, can be verified by everyone to belong to a given round.
Some level of censorship-resistance needs to be provided by the public bulletin board.
It should allow any participants who honestly follow the protocol to publish their message within the correct round.
Another functionality that is assumed, is a way for participants of a protocol to deposit funds.
These funds are later directly sent to or withdrawable by certain participants,
as decided by a protocol-defined deterministic settlement function,
which takes as input the state of the public bulletin board.
All of this functionality can be provided by a strongly programmable blockchain.
For example, it can be realized through smart contracts in Ethereum~\cite{ethereum}
or objects in Sui~\cite{sui_lutris} programmed via Sui Move~\cite{sui_move}.

In this section and the following, we consider constructions of lottery protocols that achieve all the properties defined below.

\begin{definition}[Security]
\label{def:security}
    Under some cryptographic hardness assumptions,
    it is computationally infeasible to impersonate another player or forge a winning lottery ticket.
\end{definition}

\begin{definition}[Fairness]
\label{def:fairness}
    Let $p_i$ be the expected payout of an honest participant $i$ if all participants follow the protocol.
    A fair lottery protocol ensures the following:
    Conditioned on any possible misbehavior of the other participants,
    participant $i$ is guaranteed an expected payout of at least $p_i$.
\end{definition}

For the example of a single winner and all $n$ participants having equal winning probability,
this just means that any honest participant has at least a $1/n$ chance of winning.
Importantly, this allows misbehaving participants to lose out in favor of correct participants.
This enables us to avoid the theoretical impossibility results about unbiasability.

% \begin{definition}[Single-Winner Fairness]
% \label{def:single-winner-fairness}
%     Let $n$ be the total number of participants in the lottery.
%     Any correct participant $i$ has at least a $1/n$ chance of winning the lottery.
% \end{definition}

\begin{definition}[Public Verifiability]
\label{def:public-verifiability}
    The winner of the lottery is determined as the result of a deterministic function of the data published to the bulletin board.
\end{definition}

Since this verification likely needs to be performed within a smart contract,
it is also highly relevant how computationally cheap it is to evaluate.
For example, this is a factor in deciding on the cryptographic primitives to use.

\begin{definition}[Liveness]
\label{def:liveness}
    Any honest participant can drive the protocol towards completion on their own.
\end{definition}

This notion of liveness prevents malicious participants from deadlocking an honest participant out of winning the lottery.
At the same time it prevents the use of any semi-honest party, acting as the lottery agency,
whose interaction would be required for termination.

\begin{definition}[Collateral-Freeness]
\label{def:collater-freeness}
    Participants do not need to post collateral to pay for possible penalties for misbehavior.
\end{definition}

Specifically, this means that participants need only deposit the ticket price for participating in the lottery.

\begin{lemma}[Abort Resistance]
    Given a protocol that achieves fairness (\cref{def:fairness}) and liveness (\cref{def:liveness}).
    A participant aborting this protocol at any point
    (i) may decrease but never increase their own payout
    and (ii) may never decrease the payout of any honest participant.
\end{lemma}

% \begin{itemize}
%     \item \textbf{Security:}
%         It is cryptographically impossible to impersonate another participant or forge a winning lottery ticket.
%     \item \textbf{Fairness:}
%         If $n-f$ participants are honest,
%         each honest participant has at least a $\frac{1}{n}$ chance of winning the lottery.
%     \item \textbf{Zero-Collateral:}
%         No participant needs to deposit any payments that go above the price of their lottery ticket.
%     \item \textbf{Public Verifiability:}
%         Anyone can deterministically verify the unique outcome of the lottery.
%     \item \textbf{Honest-Participant Liveness:}
%         If at least one participant is honest the lottery will run to completion.
%         In particular, the lottery protocol does not rely on a third party for either security or liveness.
%     \item \textbf{Abort Resistance:}
%     \item \textbf{Coalition Resistance:}
%         All of the above properties hold even if any number of participants collude.
% \end{itemize}

\subsection{Two-Party Lottery}
\label{sec:two-party-lottery}

For two participants, such a lottery can be implemented with the following simple commit-and-reveal protocol idea,
due to Delmonilo et al.~\cite{two_player_lotteries}:
Both participants initially publish a commitment to a random bit.
If both publicly open their respective commitment, the XOR of their bits determines the winner.
Otherwise, if one party does not open their commitment, the other party automatically wins the lottery.
This prevents the second party from strategically aborting the protocol to bias the output in their favor.
In the game-theoretically unlikely case that neither party opens,
some deterministic rule (e.g. first to register, lower id) can be used to break the tie.
Pseudocode for this protocol can be seen in \cref{alg:two-party-lottery}.
Also, \cref{fig:timeline} shows how time is split into three phases for this protocol.
Both players' commit messages have to be published before the end of the commit phase for the lottery to start.
Both players' openings to the commitment have to be published before the end of the opening phase for them to be considered during settlement.
After start of the settlement phase, either player may initiate settlement.

% \begin{enumerate}
%     \item \textbf{Commit:}
%         Alice chooses $a \leftarrow \mathbb{Z}_\lambda$, and Bob chooses $b \leftarrow \mathbb{Z}_\lambda$.
%         They publish $h(a)$ and $h(b)$ to the public bulletin board.
%     \item \textbf{Open:}
%         Alice publishes $a$ and Bob publishes $b$ to the public bulletin board.
%     \item \textbf{Settle:}
%         If either Alice or Bob did not reveal their secret before a timeout, the other party wins the lottery.
%         If neither party reveals, break tie arbitrarily (e.g. Alice wins by default).
%         Otherwise $w \coloneqq h(a || b)$ can be publicly calculated and Bob wins if and only if $w \equiv 1 \mod 2$.
% \end{enumerate}

\begin{figure}[!t]
    \footnotesize
    \begin{algorithmic}
        \State $\lambda$ \Comment{security parameter}
        \State $B$ \Comment{bulletin board / blockchain}
        \State $P$ \Comment{player IDs}
        \State $\mathsf{StartTime}, \mathsf{TimePerRound}$ \Comment{bounds for phases}
        \State $\mathsf{BuyIn}$ \Comment{buy-in amount}
        \Statex
        \Function{Init}{$ $}
            \State $P \gets \emptyset$
            \State $\mathsf{StartTime} \gets B.\mathsf{Time}()$
        \EndFunction
        \Statex
        \Function{Commit}{$id, c$}
            \State \textbf{assert}\ $\textsc{CurrentPhase}(B) = \mathsf{Commit}$
            \State \textbf{assert}\ $|P| < 2$
            \State $P \gets P \cup \{id\}$
            % \State $r \gets \{0,1\}^\lambda$
            % \State $c \gets h(r\ ||\ v)$
            \State $B.\mathsf{Publish}(id, \textsf{Commit}, c)$
            \State $B.\mathsf{DepositFrom}(id, \textsf{BuyIn})$
        \EndFunction
        \Statex
        \Function{Open}{$id, x, v$}
            \State \textbf{assert}\ $\textsc{CurrentPhase}(B) = \mathsf{Open}$
            % \If{$\textsc{CurrentPhase}(B) \neq \mathsf{Open}$}
            % $    \State \Return $\bot$
            % \EndIf
            \State $c \gets B.\mathsf{Read}(id, \textsf{Commit})$
            \State $c' \gets h(x\ ||\ v)$ \Comment{validate hash commitment}
            \State \textbf{assert}\ $c = c'$
            \State $B.\mathsf{Publish}(id, \textsf{Open}, v)$
        \EndFunction
        \Statex
        \Function{Settle}{$ $}
            \State \textbf{assert}\ $\textsc{CurrentPhase}(B) = \mathsf{Settle}$
            \State $a \gets B.\mathsf{Read}(P[0], \textsf{Open})$
            \State $b \gets B.\mathsf{Read}(P[1], \textsf{Open})$
            \If{$a = \bot$}
                \State $\mathsf{winner} \gets 1$
            \ElsIf{$b = \bot$}
                \State $\mathsf{winner} \gets 0$
            \Else
                \State $\mathsf{winner} \gets a \oplus b \mod 2$
            \EndIf
            \State $B.\mathsf{WithdrawTo}(P[\mathsf{winner}], 2 \cdot \textsf{BuyIn})$
        \EndFunction
        \Statex
        \Function{CurrentPhase}{$B$}
            \If{$B.\mathsf{Time}() \le \mathsf{StartTime} + \mathsf{TimePerRound}$}
                \State \Return $\mathsf{Commit}$
            \ElsIf{$B.\mathsf{Time}() \le \mathsf{StartTime} + 2 \cdot \mathsf{TimePerRound}$}
                \State \Return $\mathsf{Open}$
            \EndIf
            \State \Return $\mathsf{Settle}$
        \EndFunction
    \end{algorithmic}
    \caption{Two-party Lottery based on~\cite{two_player_lotteries}.}
    \label{alg:two-party-lottery}
\end{figure}

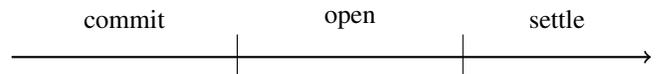
\begin{figure}[!t]
    \centering
    \begin{tikzpicture}

% Define the arrow
\draw[->, thick] (0, 0) -- (8.5, 0);

% Split the arrow into three sections and add labels
\foreach \x/\label in {1.5/commit, 4.5/open, 7.25/settle} {
    % \draw[dashed] (\x*2, 0.3) -- (\x*2, -0.3);
    \node at (\x, 0.5) {\label};
}

% Add additional lines for clarity (optional)
\draw (3, 0.3) -- (3, -0.3);
\draw (6, 0.3) -- (6, -0.3);

\end{tikzpicture}
    \caption{Timeline of the phases of the two-party lottery.}
    \label{fig:timeline}
\end{figure}

Each of the two participants ensures winning at least 50\% of the time, simply by following the protocol.
As a consequence, neither player can be better off by deviating from the protocol.
Note, that this protocol can straightforwardly be adapted to different winning probabilities.

\subsection{Single-Winner Multiparty Lottery}
\label{sec:single-winner-lottery}

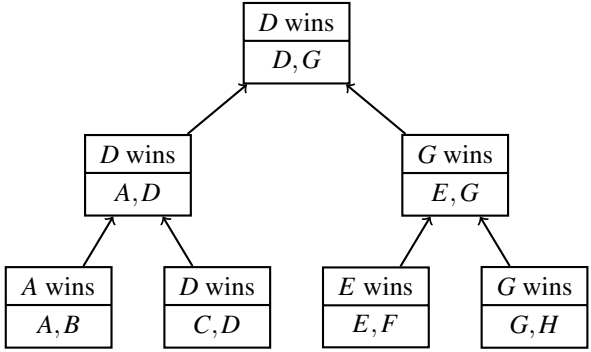
\begin{figure}[!t]
    \centering
    \begin{tikzpicture}[
    level distance=5em,
    every node/.style = {rectangle split, rectangle split parts=2, minimum width=4em, draw, align=center},
    level 1/.style={sibling distance=12em},
    level 2/.style={sibling distance=6em},
    thick,<-]

    \node{$D$ wins\nodepart{two}$D, G$}
    child { node {$D$ wins\nodepart{two}$A, D$}
        child { node {$A$ wins\nodepart{two}$A, B$} }
        child { node {$D$ wins\nodepart{two}$C, D$} }
    }
    child { node {$G$ wins\nodepart{two}$E, G$}
        child { node {$E$ wins\nodepart{two}$E, F$} }
        child { node {$G$ wins\nodepart{two}$G, H$} }
    };
\end{tikzpicture}
    \caption{Example tournament tree with 8 players for the single-winner lottery based on~\cite{zero_collateral_lotteries}.
        Each node in the tree represents one instance of the two-party lottery.
        All winners of any layer advance to the next higher layer.}
    \label{fig:tree-1-winner}
\end{figure}

As previously shown by Miller and Bentov~\cite{zero_collateral_lotteries},
the above two-player protocol (cf. \cref{sec:two-party-lottery}) can be used as a sub-protocol for constructing a multiparty lottery with a single winner.
This can be naturally done when the number of players $n$ is a power of two,
by building a perfect binary tree of two-player instances.
In each round the winners advance to the next level.
The single participant to win in all $\log(n)$ consecutive rounds is the winner of the multiparty lottery.
An example of this can be seen in \cref{fig:tree-1-winner}.

\begin{figure*}[!t]
    \centering
    \begin{tikzpicture}

% Define the arrow
\draw[->, thick] (0, 0) -- (17.5, 0);

% Split the arrow into three sections and add labels
\foreach \x/\label in {1.75/commit, 5.25/open(0), 8.75/settle(0) \& open(1), 12.25/settle(1) \& open(2), 15.75/settle} {
    % \draw[dashed] (\x*2, 0.3) -- (\x*2, -0.3);
    \node at (\x, 0.5) {\label};
}

% Add additional lines for clarity (optional)
\draw (3.5, 0.3) -- (3.5, -0.3);
\draw (7, 0.3) -- (7, -0.3);
\draw (10.5, 0.3) -- (10.5, -0.3);
\draw (14, 0.3) -- (14, -0.3);

\end{tikzpicture}
    \caption{Phases for multiparty lottery.
        Example for tournament tree height of three, i.e., $n=8$ participants, such as the one seen in \cref{fig:tree-1-winner}.
        Some phases serve the function of two phases in the underlying two-party lottery protocols and are denoted as such.}
    \label{fig:timeline-multiparty}
\end{figure*}
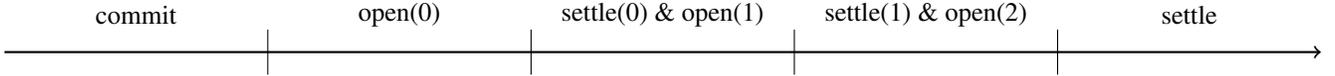

Recall that a two-player lottery requires three sequential rounds (cf. \cref{fig:timeline}).
Still, in practice, the single-winner multiparty lottery consisting of $\log(n)$ such lotteries in sequence does not require $3\log(n)$ rounds.
Instead, it can be implemented with only $\log(n)+3$ sequential rounds of communication with their own respective timeouts.
First of all, values for all rounds can be committed to upfront in a single commit round.
One obvious implementation would be vector commitments~\cite{vector_commitments}, like Merkle trees~\cite{merkle_tree}.
However, due to the sequential nature of the protocol, there is a more efficient solution.
By initially committing with a long enough chain of (hash) commitments, similar to Lamport's one-time passwords~\cite{lamport_otp},
this can even be achieved with (i) one constant sized commit value in the commit round
and (ii) one constant sized opening in each opening round.
This way each opening is implicitly also a commitment for the next round.
Furthermore, almost all settlement and opening rounds can be merged.
For this the winner settles their round $i-1$ two-party lottery when opening their commitment for the one in round $i$.
In a smart contract implementation, this could be allowed to happen as a single transaction.
The resulting timeline can be seen in \cref{fig:timeline-multiparty}.
As an additional optimization for the good case,
any round can be terminated before its timeout as soon as the last player published a valid message for that round.

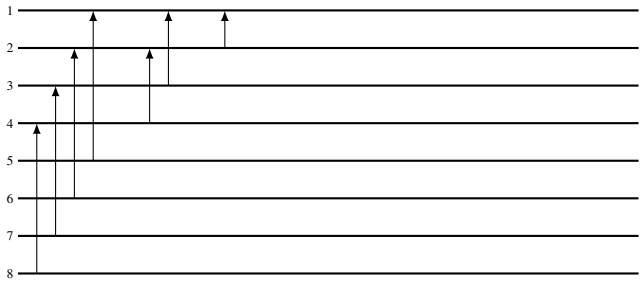
\begin{figure}[!t]
    \centering
    \begin{tikzpicture}[scale=0.5, transform shape]
    % Define styles for comparators and wires
    \tikzstyle{wire}=[thick]
    \tikzstyle{comparator}=[draw, -latex]

    % Draw wires
    \foreach \i in {1,...,8} {
        \draw[wire] (0,\i) -- (16.5,\i);
    }
  
    % Labels
    \foreach \i in {1,...,8} {
        \node[left] at (0,9-\i) {\i};
    }
    
    % Draw depth-1 comparators
    \foreach \i in {1,2,3,4} {
        \draw[comparator] (0.5*\i,\i) -- (0.5*\i,\i+4);
    }
    
    % Draw depth-2 comparators
    \foreach \i in {5} {
        \draw[comparator] (3.5,\i) -- (3.5,\i+2);
    }
    \foreach \i in {6} {
        \draw[comparator] (4,\i) -- (4,\i+2);
    }

    % Draw depth-3 comparators
    \foreach \i in {7} {
        \draw[comparator] (5.5,\i) -- (5.5,\i+1);
    }
\end{tikzpicture}
    \caption{Same tournament as in \cref{fig:tree-1-winner},
        depicted in sorting network~\cite{sorting_networks} notation as introduced by Donald Knuth~\cite{knuth_notation}.
        Each arrow indicates a single two-player coin flip with 50\% chance of either player winning,
        with the arrow head indicating the new position of the winner.}
    \label{fig:sorter-1-winner}
\end{figure}

\section{Main \sysname Protocol}
\label{sec:main-protocol}

\begin{table*}[!t]
    \centering
    \begin{tabular}[t]{lccccccc}
        \toprule
        \textbf{Protocol} & \textbf{Section} & \textbf{Weighted} & \textbf{Multi-Winner} & \textbf{Unlinkability} & \textbf{Rounds} & \textbf{Messages} & \textbf{Computation} \\
        \midrule
        Miller and Bentov \cite{zero_collateral_lotteries} & \ref{sec:single-winner-lottery} & --- & --- & ---    & $\log n$   & $n$          & $n$          \\
        single-winner \cite{purelottery}     & \ref{sec:weighted-winning-probabilities} & yes & --- & ---    & $\log n$   & $n$          & $n$          \\
        sequential~\cite{purelottery_thesis} & \ref{sec:naive-generalizations}          & yes & yes & ---    & $n \log n$    & $n^2$        & $n^2$        \\
        parallel~\cite{purelottery_thesis}   & \ref{sec:naive-generalizations}          & yes & yes & ---    & $\log n$   & $n^2 \log n$ & $n^2 \log n$ \\
        \midrule
        parallel + PRNG                      & \ref{sec:naive-generalizations}          & yes & yes & ---    & $\log n$   & $n \log n$   & $n^2$        \\
        \sysname                             & \ref{sec:arbitrary-payout-functions}     & --- & yes & ---    & $\log n$   & $n \log n$   & $n \log n$   \\
        perfect shuffling                    & \ref{sec:perfect-shuffling}              & yes & yes & ---    & $n$            & $n^2$        & $n^2$        \\
        \sysnamezkp                          & \ref{sec:unlinkable-lotteries}           & --- & yes & strong & $\log n$   & $n \log n$   & $n \log n$   \\
        \sysnamewul                          & \ref{sec:unlinkable-lotteries}           & --- & yes & weak   & $\log n$   & $n \log n$   & $n \log n$   \\
        \bottomrule
    \end{tabular}
    \caption{Comparison of different multiparty lottery protocols regarding the additional properties achieved
        and the asymptotic complexity in terms of sequential rounds, total message sizes, and computation time.
        Specifically comparing previous work and the constructions seen in \cref{sec:naive-generalizations} to the \sysname family.}
    \label{tab:asymptotic-complexity}
\end{table*}

The previous protocol does not natively support arbitrary (non-power of two) $n$,
arbitrary winning probabilities and arbitrary payout functions.
In this section we describe how to construct more generally applicable multiparty lottery protocols.
The additional goals of this lottery are captured by the following definitions.

\begin{definition}[Payout Function]
    %Assuming we want to pay a price to more than a single winner, i.e.,
    A payout function is a monotonic non-increasing function
    $p : [n] \rightarrow [0; 1]_\mathbb{R}$ with $\sum_{i=1}^n p(i) = 1$,
    which determines which fraction of the entire pot each participant wins based on their position in the output permutation.
\end{definition}

%Without loss of generality, we consider only the case of monotonic non-increasing payout functions.
The single-winner case is then just the special case where, for any $n$,
the payout function is $p_n(1) = 1$ and $p_n(i) = 0$ for any $i \in [n] \setminus \{1\}$.
That is, it can be considered as a specific family of payout functions.
All other payout functions are called the multi-winner case.

Solving the multi-winner case in general is possible by generating all permutations with certain probabilities.
However, for the $k$-winner case, as an optimization it might be more efficient to only select $k$ participants.
For example, one could run $k$ instances of the single-winner protocol to achieve that.

We define winning probabilities in fully-honest executions of the multi-winner case as being equivalent to
iterative rounds of single winner lotteries, with $n, n-1, \ldots, 2$ participants.
First place is defined exactly as the winner in the single-winner case.
Other positions are defined recursively by removing the previous winner and applying again.
With equal winning probabilities for all players in the single-winner case,
this is naturally equivalent to all permutations being equally likely outcomes.

\begin{definition}[Weights]
    During registration each participant $i$ is assigned a weight $w_i \in \mathbb{R}^{+}$.
\end{definition}

Weights may be assigned by the protocol according to some arbitrary application-specific rule.

With weights, the winning probability of any participant in the single-winner case then canonically generalizes.
Let $n$ be the total number of participants in the lottery with weights $w_1, \ldots, w_n$.
In a fully-honest execution of the lottery,
participant $i$ has a $\frac{w_i}{w_1 + \ldots + w_n}$ chance of winning the lottery.

Since the above description of multi-winner winning probabilities is independent
of the specific rule for single-winner winning probabilities it is based on,
it can also be used in conjunction with weights.
Fairness (\cref{def:fairness}) applies straightforwardly to any of the above cases of winning probabilities.

Concretely, in the following, we propose various protocols.
An overview of all newly proposed protocols compared to previous work can be found in \cref{tab:asymptotic-complexity}.

\subsection{Weighted Winning Probabilities}
\label{sec:weighted-winning-probabilities}

Next we show a generalization of the single-winner lottery from \cref{sec:single-winner-lottery} to weighted winning probabilities.
For example, weights could simply be based on the buy-in price of the lottery ticket.
Furthermore, this change to the protocol inherently makes it support numbers of participants that are not a perfect power-of-two.
This is naturally accommodated for when all participants are initially assigned weight one,
and empty spots (filling up to a power-of-two) are assigned weight zero.

\begin{figure}[!t]
    \centering
    \begin{tikzpicture}[
    level distance=5em,
    every node/.style = {rectangle split, rectangle split parts=2, text height=1em, minimum width=4em, draw, align=center},
    level 1/.style={sibling distance=12em},
    level 2/.style={sibling distance=6em},
    thick,<-]

    \node{$\frac{5}{13}, \frac{8}{13}$\nodepart{two}$5, 8$}
    child { node {$\frac{2}{5}, \frac{3}{5}$\nodepart{two}$2, 3$}
        child { node {$\frac{1}{2}, \frac{1}{2}$\nodepart{two}$1, 1$} }
        child { node {$\frac{2}{3}, \frac{1}{3}$\nodepart{two}$2, 1$} }
    }
    child { node {$\frac{7}{8}, \frac{1}{8}$\nodepart{two}$7, 1$}
        child { node {$\frac{4}{7}, \frac{3}{7}$\nodepart{two}$4, 3$} }
        child { node {$1, 0$\nodepart{two}$1, \bot$} }
    };
\end{tikzpicture}
    \caption{Example of a weighted single-winner lottery with 7 players,
        as also realized in~\cite{purelottery}.
        Each node in the tree represents one instance of the two-party lottery.
        Bottom numbers represent the weights of left and right player respectively,
        top numbers represent their winning probabilities.}
    \label{fig:tree-1-winner-weights}
\end{figure}
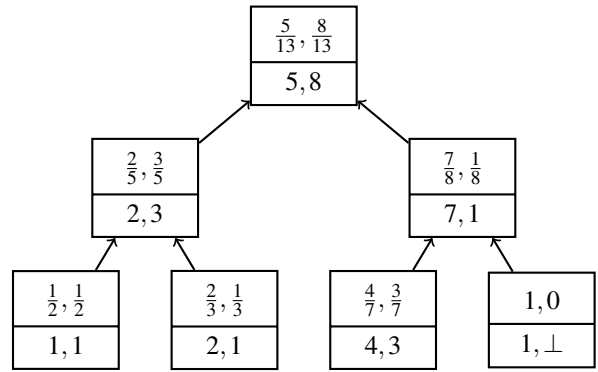

Augmenting the single-winner lottery with weights is rather straightforward and has been previously considered by Ballweg et al.~\cite{purelottery}.
Each node of the tournament tree is assigned a weight.
Initially the leaves are assigned the participants weights, assigned to them when entering the lottery.
For any two-party lottery where participants with weights $w_1$ and $w_2$ are competing against each other,
their respective winning probabilities are $\frac{w_1}{w_1 + w_2}$ and $\frac{w_2}{w_1 + w_2}$.
A winner of any two-party lottery is assigned the sum of the two participants' weights in the next round.
A practical example of this entire tournament structure can be seen in \cref{fig:tree-1-winner-weights}.

Concretely, determining the winner with the correct probabilities may be implemented as follows.
Call the values revealed by the two participants $v_1$ and $v_2$.
We interpret $x \coloneqq v_1 \oplus v_2$ as a $\lambda$-bit fixed-precision rational number in the interval $[0; 1)$.
If $x < \frac{w_1}{w_1 + w_2}$, participant 1 wins, otherwise participant 2 wins.

\subsection{Naive Generalizations}
\label{sec:naive-generalizations}

The tree-based single-winner protocols of \cref{sec:single-winner-lottery,sec:weighted-winning-probabilities} can easily,
albeit inefficiently, be extended to a full ranking of the $n$ participants.
This can be done by running multiple instances of the protocol in sequence.
First it is run by all $n$ participants for determining first place,
then the same protocol is run by the remaining $n-1$ participants to determine second place, and so on.

However, this blows up the communication complexity as well as the number of rounds of the protocol each by a factor of $\Theta(n)$,
bringing them to $\Theta(n^2)$ and $\Theta(n \log n)$ respectively.
The number of rounds (but not communication complexity) can be reduced by using rejection sampling instead.
That is, instead of running $n$ instances of the single-winner lottery in sequence,
one can run $\Theta(n \log n)$ instances of it in parallel~\cite{coupon_collector_problem}.
Instances selecting a previously selected participant are simply disregarded for the final ranking.
In fact, this approach further increases communication complexity, bringing it up to $\Theta(n^2 \log n)$.

Another option is simulating sequential execution in retrospect.
This way, the required number of instances can be reduced from $\Theta(n \log n)$ to the optimal of at most $n$.
Later instances (in the order of outcome positions they determine) can be conditioned on the set of participants remaining.
This can only be done in hindsight (after evaluating previous instances).
However, in contrast to naive sequential instances, this does not require additional rounds of communication.

If we assume existence of a pseudo-random number generator,
the instances do not even each need their own values from different commitments to be revealed.
Instead, the participants can commit to and reveal a sequence of high-entropy seed values during the one live instance.
From each seed value up to $n$ random values are then derived during the instances simulated in retrospect.

How a value is derived for instance $i$ from the seed should never depend on the outcome of any other instances.
This is important in order to prevent any new biasing attack vectors from opening up.
For example, the opening should be derived as $h(s \| i)$.
Further, it is important that openings for the simulated rounds
are derived based on the opening from the corresponding round in the live instance.
Otherwise, abort resistance only holds for the lottery for first place.
In practice, this requires either storing all openings for all participants until the entire lottery is settled
or (if using hash-chain commitments) recomputing the openings for simulated rounds from the last opening.

\subsection{Arbitrary Payout Functions}
\label{sec:arbitrary-payout-functions}

The trivial way of turning a single-winner lottery into a multi-winner lottery is by running $n$ instances of the single-winner lottery sequentially.
This results in a total order of all participants, which then serves as input to the payout function.
Using the single-winner lottery tournament tree for each instance would thus result in a total runtime of $\Theta(n \log n)$ and a total of $\Theta(n^2)$ messages.
Our construction of the standard \sysname protocol, presented in this section,
is more efficient but assumes equal winning probability for all participants.
So, we first argue that this assumption is reasonable.

\begin{theorem}[Split Dominance]
\label{thm:split-dominance}
    Let there be a multi-winner lottery with $n$ participants,
    payout function $p_n : [n] \rightarrow \mathbb{R}$,
    and minimum buy-in $w_\mathsf{min}$.
    Participants are able to enter the lottery with buy-in $w = k w_\mathsf{min}$ for any $k \in \mathbb{N}$.
    Assume that transaction fees for participating in the lottery are zero.
    Entering the lottery with $k$ individual entries of $w_\mathsf{min}$ each, is a weakly dominating strategy.
    For any payout function $p_n$, that is not single-winner, it is even strictly dominating.
\end{theorem}

\begin{proof}
    For clarity, consider the equivalence to the naive sequential single-winner lotteries.
    In each single-winner lottery with participating weights $w_1, \ldots, w_n$,
    the chance for a user controlling total stake $w$ to win,
    is $\frac{w}{w_1 + \ldots + w_n}$.
    It does not depend on how the weight is split between one or multiple participants.
    Thus, the chance for winning a single-winner lottery round does not change when splitting or merging buy-ins.
    However, the winning entry will be removed from the following round.
    So removing the least amount of weight per round maximizes future chances of winning.

    Specifically, assume the user has an entry with weight $w_1 > w_\textsf{min}$.
    If they win the first single-winner lottery with this entry,
    their winning probability in the second one is
    $\frac{w-w_1}{w_1 + \ldots + w_n} < \frac{w-w_\textsf{min}}{w_1 + \ldots + w_n}$,
    Given $p(2) > 0$, their expected payout for that round is then
    $\frac{w-w_1}{w_1 + \ldots + w_n} p(2) < \frac{w-w_\textsf{min}}{w_1 + \ldots + w_n} p(2)$.
    So, given $p(1), p(2) > 0$,
    having no entry larger than $w_\textsf{min}$ that could be split is a dominating strategy.
    By the same argument, there is no disadvantage even if $p(2) = 0$.
\end{proof}

\begin{corollary}
    In a single-winner lottery, assuming transaction fees are strictly positive.
    A user always has positive expected value for merging their entries.
    The difference in expected value is exactly equal to the saved transaction fees.
\end{corollary}

\begin{corollary}
    In practice, $w_\textsf{min} \cdot p(2)$ has to be larger than the transactions fees for a participant to have a strictly positive change in expected payout for splitting their entries.
\end{corollary}

Based on these results one can argue that the two most interesting cases of multiparty lotteries are
(i) single-winner lotteries with weighted winning probabilities
and (ii) multi-winner lotteries with equal winning probabilities for all participants.
For the first case, the construction seen in \cref{sec:weighted-winning-probabilities} is almost ideal.
So, next we focus on the second case.

First, we define shuffling networks, monotonic shuffling networks.
The, we show how monotonic shuffling networks can generally yield
fair constructions when instantiated with the two-party lottery protocol.
Finally, we give a concrete protocol, called \sysname.

Shuffling networks are defined analogous to sorting networks~\cite{sorting_networks}.
We also use the same visual representation often used for sorting networks,
which was originally introduced by Donald Knuth~\cite{knuth_notation}.
However, we always depict shufflers as directed arrows.
This is relevant because of the asymmetry of the two-party lottery protocol used to instantiate them.

\begin{definition}[Shuffling Network]
    \label{def:shuffling-network}
    A shuffling network is defined by a number of wires $n$ and a sequence of swappers $(i, j, p)$,
    which swap wires $i$ and $j$ with probability $p$.
    We call a shuffling network \emph{correct} if any of the $n$ inputs has probability $1/n$ of appearing on any specific output wire.
    We call a shuffling network \emph{perfect} if any permutation of the inputs is equally likely.
\end{definition}

A shuffling network being perfect trivially implies it being correct.
Shuffling networks can additionally be interpreted as directed.
For example, say canonically that a swapper $(i, j, p)$ is directed from $i$ to $j$.

\begin{definition}[Monotonic Shuffling Network]
    \label{def:monotonic-shuffling-network}
    We say a shuffling network is \emph{monotonic} if the following holds:
    A shuffling network is monotonic if, for every swapper $(i,j,p)$ in the network,
    the change in expected payout from swapping a wire from position $i$ to position $j$ is non-negative,
    given a payout function is applied to the output positions in canonical order.
\end{definition}

This definition is essential for fairness,
when instantiating swappers with the two-party lottery protocol,
considering the one-sided biasability of that protocol.

\begin{theorem}[Fair Shuffling]
    Instantiate each swapper of a correct monotonic shuffling network with the two-party lottery from \cref{sec:two-party-lottery}.
    Ensure that the winner of the two-party lottery moves towards the higher expected payout.
    The resulting protocol achieves fairness according to \cref{def:fairness} for any monotonic non-increasing payout function.
\end{theorem}

\begin{proof}
    From the definition of monotonic shuffling networks (\cref{def:monotonic-shuffling-network}),
    it directly follows that winning the two-party lottery for any specific swapper may increase but never decrease the expected payout.
    Because the exact payout depends on the values revealed by future opponents,
    expected payout is all the decision can be based on.
\end{proof}

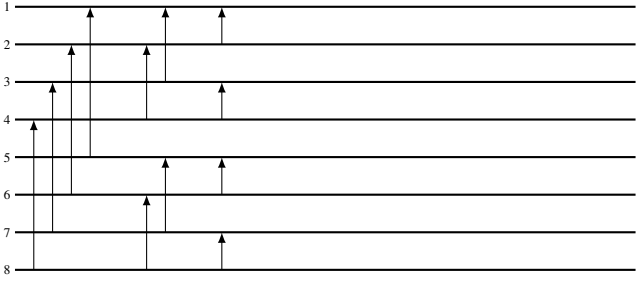
\begin{figure}[!t]
    \centering
    \begin{tikzpicture}[scale=0.5, transform shape]
    % Define styles for comparators and wires
    \tikzstyle{wire}=[thick]
    \tikzstyle{comparator}=[draw, -latex]

    % Draw wires
    \foreach \i in {1,...,8} {
        \draw[wire] (0,\i) -- (16.5,\i);
    }
  
    % Labels
    \foreach \i in {1,...,8} {
        \node[left] at (0,9-\i) {\i};
    }
    
    % Draw depth-1 comparators
    \foreach \i in {1,2,3,4} {
        \draw[comparator] (0.5*\i,\i) -- (0.5*\i,\i+4);
    }
    
    % Draw depth-2 comparators
    \foreach \i in {1,5} {
        \draw[comparator] (3.5,\i) -- (3.5,\i+2);
    }
    \foreach \i in {2,6} {
        \draw[comparator] (4,\i) -- (4,\i+2);
    }

    % Draw depth-3 comparators
    \foreach \i in {1,3,5,7} {
        \draw[comparator] (5.5,\i) -- (5.5,\i+1);
    }
\end{tikzpicture}
    \caption{Visual representation of a correct shuffling network,
        that corresponds to a multi-winner lottery with 8 participants.
        This construction generalizes to any power of two.
        Each arrow indicates a single two-party coin flip.}
    \label{fig:sorter-8}
\end{figure}

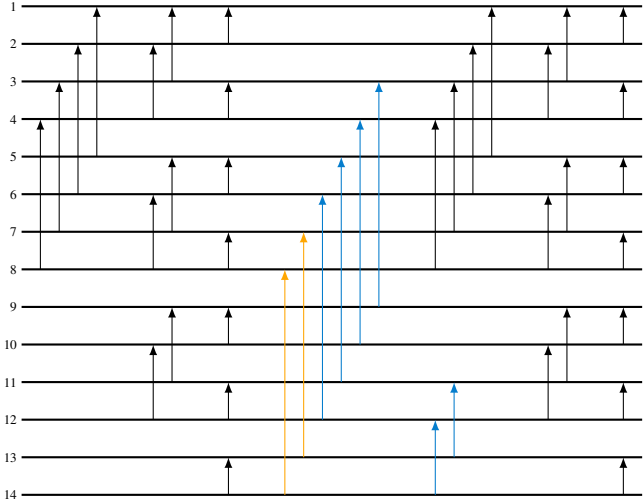
\begin{figure}[!t]
    \centering
    \begin{tikzpicture}[scale=0.5, transform shape]
    % Define styles for comparators and wires
    \tikzstyle{wire}=[thick]
    \tikzstyle{comparator}=[draw, -latex]

    % Draw wires
    \foreach \i in {1,...,14} {
        \draw[wire] (0,\i) -- (16.5,\i);
    }
  
    % Labels
    \foreach \i in {1,...,14} {
        \node[left] at (0,15-\i) {\i};
    }

    % Draw depth-1 comparators
    \foreach \i in {7,8,9,10} {
        \draw[comparator] (0.5*\i-3,\i) -- (0.5*\i-3,\i+4);
    }
    
    % Draw depth-2 comparators
    \foreach \i in {3,7,11} {
        \draw[comparator] (3.5,\i) -- (3.5,\i+2);
    }
    \foreach \i in {4,8,12} {
        \draw[comparator] (4,\i) -- (4,\i+2);
    }

    % Draw depth-3 comparators
    \foreach \i in {1,3,5,7,9,11,13} {
        \draw[comparator] (5.5,\i) -- (5.5,\i+1);
    }

    % Draw merge lines
    \foreach \i in {1,2} {
        \draw[comparator, darkyellow] (6.5+0.5*\i,\i) -- (6.5+0.5*\i,\i+6);
    }
    \foreach \i in {3,...,6} {
        \draw[comparator, darkblue] (6.5+0.5*\i,\i) -- (6.5+0.5*\i,\i+6);
    }
    
    \foreach \i in {1,2} {
        \draw[comparator, darkblue] (10.5+0.5*\i,\i) -- (10.5+0.5*\i,\i+2);
    }

    % Draw depth-1 comparators
    \foreach \i in {7,8,9,10} {
        \draw[comparator] (10.5+0.5*\i-3,\i) -- (10.5+0.5*\i-3,\i+4);
    }
    
    % Draw depth-2 comparators
    \foreach \i in {3,7,11} {
        \draw[comparator] (14,\i) -- (14,\i+2);
    }
    \foreach \i in {4,8,12} {
        \draw[comparator] (14.5,\i) -- (14.5,\i+2);
    }

    % Draw depth-3 comparators
    \foreach \i in {1,3,5,7,9,11,13} {
        \draw[comparator] (16,\i) -- (16,\i+1);
    }
\end{tikzpicture}
    \caption{Full example of \sysname correct shuffling network with 14 participants.
        Each arrow indicates a single two-party lottery.
        The probability for the participants to swap positions is indicated by the arrow color
        (black: $\frac{1}{2}$, blue: $\frac{2}{3}$, yellow: $\frac{4}{5}$).}
    \label{fig:full-example}
\end{figure}

There is a general construction of correct shuffling networks for arbitrary $n$ with depth at most $2 \log(n) + 1$.
First we give a construction from $p=1/2$ swappers for $n=2^k$ for some $k \in \mathbb{N}$.
Secondly, we show how these can be combined with some additional $p \neq 1/2$ swappers to generalize to arbitrary $n$.

The construction seen in \cref{fig:sorter-8} generalizes to any power of two $n=2^k$.
In general it is defined by the sequence of rounds $r = (0, \ldots, k-1)$ with swappers:

\[ \bigcup_{b=1}^{2^r} \bigcup_{i=0}^{b-1} \{(b \cdot 2^{k-r}-i, b \cdot 2^{k-r-1}-i, \frac{1}{2})\} \]

That makes it a shuffling network of depth $k$ with $2^{k-1}$ swappers at each depth,
for a total of $k \cdot 2^{k-1}$ swappers.
We call these $S_1, S_2, S_4, \ldots$ and will use them as building blocks for the general construction below.
Correctness of these follows directly from the recursive construction.

For an arbitrary $n \in \mathbb{N}$, let $k_1 < k_2 < \cdots < k_l$ be its unique power-of-two decomposition,
i.e., $n = 2^{k_1} + 2^{k_2} + \ldots + 2^{k_l}$.
Consider the wires as batches of $k_1, k_2, \ldots, k_l$ wires.
Firstly, we apply $S_{2^{k_1}}$ to the first batch, $S_{2^{k_2}}$ to the second one, and so on.
Secondly, we merge each batch with all smaller ones.
This is formalized below.
Lastly, we apply the power-of-two shuffling networks to their respective batches again.

Merging a batch $k_i$ (for some $i \in [l]$) with smaller batches
is realized by a single round with the following swappers:

\[ \bigcup_{j=1}^{i-1} \bigcup_{o=0}^{2^{k_j}-1} \{(n - o - \sum_{t=1}^{i-1} 2^{k_t} , 2^{k_i} - o - \sum_{t=1}^{i-1} 2^{k_t}, 1 - \frac{2^{k_j}}{2^{k_i}})\} \]

It is obvious from the construction that the network has depth at most $3 k_l$.
However, it actually only has depth $2 k_l + 1$, as seen in \cref{fig:full-example}.
Importantly, the above construction is also carefully crafted to be monotonic.
This allows it to be used as the tournament structure in the \sysname protocol that will be explained afterwards.

\begin{theorem}[General Shuffling Networks]
    For any $n \in \mathbb{N}$, the above construction gives a correct monotonic shuffling network with depth $2\floor{\log(n)} + 1$.
\end{theorem}

\begin{proof}
    Follows from the three lemmas below.
\end{proof}

\begin{lemma}[Correctness]
    For any $n \in \mathbb{N}$, the above construction gives a correct shuffling network.
\end{lemma}

\begin{proof}
    Let $k_1 < k_2 < \cdots < k_l$ be the unique power-of-two decomposition of $n$,
    i.e., $n = 2^{k_1} + 2^{k_2} + \ldots + 2^{k_l}$.
    After applying $S_{2^{k_i}}$ to batch $i$ for the first time,
    each input is equally likely ($2^{-k_i}$) to be on each of the batches wires.

    Probabilities of the swappers in the merging rounds are carefully chosen,
    such that it ensures that afterwards any batch $i$ has the correct probability density
    (i.e., proportional to the number of wires within the batch) of each input wire appearing.

    Once each batch $i$ has the correct probability mass,
    the second application of shuffling network $S_{2^{k_i}}$ distributes it evenly across all wires within the batch.
    This follows directly from correctness of $S_{2^{k_i}}$.
\end{proof}

\begin{lemma}[Shallow Depth]
    For any $n \in \mathbb{N}$, the above construction gives a shuffling network with depth $2\floor{\log(n)} + 1$.
\end{lemma}

\begin{proof}
    Looking at the largest batch, we have $S_{2^{k_l}}$ two times,
    with depth $k_l$ each, and an additional round for merging.
    So, the entire network already has depth at least $2 k_l + 1$.
    Consider any other $k_i$ (for any $i \in [l-1]$).
    It is also involved in the merging of the largest batch.
    Before that it has $S_{2^{k_i}}$ of depth $k_i < k_l$.
    After the merge round it has up to $k_l-k_i$ more merge rounds and another $S_{2^{k_i}}$,
    for a combined depth of at most $k_l$.
    Thus, the entire network has depth exactly $2 k_l + 1$.
\end{proof}

\begin{lemma}[Monotonicity]
    For any $n \in \mathbb{N}$, the above construction gives a monotonic shuffling network.
\end{lemma}

\begin{proof}
    Monotonicity of the power-of-two sub-networks follows from their recursive structure:
    For each round, they move winners into the better half,
    losers into the worse half, and then treat these separately.
    Since monotonicity is a composable property, it only remains to show that the merging rounds preserve monotonicity.
    
    Consider the merging round for batch $i$.
    All batches which are lower (in the monotonic order of expected payout)
    are only potentially swapped with lower wires in batch $i$.
    For any two batches being merged,
    all wires which are lower in one batch are only potentially swapped with lower wires in the other batch.
    Therefore, before the merge rounds expected payouts within each batch are monotonic.
\end{proof}

The above construction is not minimal in the number of swappers. 
For some $n$ (i.e., for any $n$ that is not of the form $2^k-1$)
there are some ``no-op'' swappers, where swapping does not change expected payout for either party.
These swappers can be pruned when initially constructing the tournament tree.
Also, depending on the payout function, certain games can be pruned.
Specifically, if for some $i^{*}$, we have $p(i) = 0$ for all $i \in \{i^{*}+1, \ldots, n\}$,
then any game not relevant to the ranking of the top $i^{*}$ outputs (i.e. in the ``bottom right'' of the network) can be ignored.
In practice participants have no incentive to pay transaction fees to participate in these games anyways.

In \cref{alg:multi-winner-lottery} you can see the general \sysname multiparty lottery protocol.
It can be instantiated with any correct monotonic shuffling network as its tournament structure.
For example, it could also be applied to the basic tournament tree shuffling network \cref{fig:sorter-1-winner},
in order to implement the single winner lottery from \cref{sec:single-winner-lottery}.

Note that, when using a \emph{correct} but not \emph{perfect} shuffling network,
outcomes of this protocol are dependent on the initial order of the participants.
While all positions have the same expected value, some combinations of positions have higher variance.
This is relevant if we assume the setting where participants are able to enter the lottery with more than one ticket.
In this case a risk averse user may try to pick their ticket positions in a way that minimizes variance.

The effectiveness of this strategy can be mitigated by adding a single commit-and-reveal round at the beginning involving all participants.
Values revealed in this round determine the initial permutation of participants.
Just as in any two-party round, not revealing a value causes disqualification.
To bias this decision the adversary thus has to reduce their expected payout.
So, for an almost risk-neutral adversary this defeats the strategy.

In case we do want to handle the combined multi-winner case with weighted winning probabilities with,
we may generalize the notion of shuffling networks to account for weights.

\begin{definition}[Weighted Shuffling Network]
    \label{def:weighted-shuffling-network}
    A \emph{weighted} shuffling network is a shuffling network where
    (i) each input wire $i$ has a weight $w_i$ associated with it,
    and (ii) each swapper has an associated probability function $p(w_i, w_j)$ instead of a fixed probability.
    Also, each swapper has an associated weight aggregation function,
    so, the winner's wire is newly assigned the aggregate of the incoming weights $a(w_i, w_j)$,
    whereas the loser's wire retains their previous weight.
\end{definition}

Constructing correct weighted shuffling networks,
just as constructing perfect shuffling networks, even for the unweighted case,
is not straightforward.
Both of these problems are outside the scope of this work but are interesting future work.
Given the \sysname protocol they could immediately be turned into practical lotteries.

% \begin{theorem}[Correct Weighted Shuffling Networks]
%     For any $n \in \mathbb{N}$ the above construction gives a correct weighted shuffling network of depth $n$.
% \end{theorem}
% 
% \begin{proof}
%     
% \end{proof}
% 
% We also give a general construction for turning any \emph{correct} $(n, n)$ (unweighted) shuffling network into a \emph{perfect} $(n, k)$ weighted shuffling network using $k$ additional rounds.
% 
% \begin{theorem}[Perfect Weighted Shuffling Networks]
%     For any $n \in \mathbb{N}$ the above construction transforms a correct unweighted $(n, n)$ shuffling network of depth $d$ into a perfect weighted $(n, k)$ shuffling network of depth $d + k$.
% \end{theorem}
% 
% \begin{proof}
%     
% \end{proof}
% 
% The variant with arbitrary winning probabilities and payout functions works as follows (commit and reveal phases stay the same):

%An advantage of the hash-chain commitment is that it allows to skip on-chain opening for rounds without aborts.

\begin{figure}[!t]
    \footnotesize
    \begin{algorithmic}
        \State $\lambda$ \Comment{security parameter}
        \State $B$ \Comment{bulletin board / blockchain}
        \State $P\ (n \coloneqq |P|,\ k \coloneqq \floor{\log n})$ \Comment{player IDs}
        \State $T$ \Comment{tournament structure (array of rounds, each an array of swappers)}
        \State $p_n$ \Comment{payout function}
        \State $\mathsf{StartTime}, \mathsf{TimePerRound}$ \Comment{bounds for phases}
        \State $\mathsf{BuyIn}$ \Comment{buy-in amount}
        \Statex
        \Function{Commit}{$id, c$}
            \State \textbf{assert}\ $\mathsf{Commit} \in \textsc{CurrentPhase}(B)$
            \State $P \gets P \cup \{id\}$
            \State $B.\mathsf{Publish}(id, \textsf{Commit}(0), c)$
            \State $B.\mathsf{DepositFrom}(id, \textsf{BuyIn})$
        \EndFunction
        \Statex
        \Function{Open}{$id, r, c_\mathsf{new}, v$}
            \State \textbf{assert}\ $\mathsf{Open}(i) \in \textsc{CurrentPhase}(B)$
            \State $c \gets B.\mathsf{Read}(id, \mathsf{Commit}(r))$
            \State \textbf{assert}\ $c \neq \bot$
            \State $c' \gets h(c_\mathsf{new}\ ||\ v)$
            \State \textbf{assert}\ $c = c'$
            \State $B.\mathsf{Publish}(id, \mathsf{Open}(r), v)$
            \State $B.\mathsf{Publish}(id, \mathsf{Commit}(r+1), c_\mathsf{new})$
        \EndFunction
        \Statex
        \Function{SettleMatch}{$r, m$}
            \State $(i, j, p) \gets T[r][m]$
            \State \textbf{assert}\ $\mathsf{Settle}(i) \in \textsc{CurrentPhase}(B)$
            \State $a \gets B.\mathsf{Read}(i, \mathsf{Open}(r))$
            \State $b \gets B.\mathsf{Read}(j, \mathsf{Open}(r))$
            \If{$a = \bot$}
                \State $\textsf{winner} \gets j$
            \ElsIf{$b = \bot$}
                \State $\textsf{winner} \gets i$
            \ElsIf{$\textsf{FixedPointReal}(a \oplus b) < \frac{P[i].weight}{P[i].weight + P[j].weight}$}
                \State $\textsf{winner} \gets i$
            \Else
                \State $\textsf{winner} \gets j$
            \EndIf
            \If{$\mathsf{winner} = i$}
                \State $P[i], P[j] \gets P[j], P[i]$ \Comment{swap players}
            \EndIf
        \EndFunction
        \Statex
        \Function{SettleLottery}{$\ $}
            \State \textbf{assert} $\mathsf{Settle} \in \textsc{CurrentPhase}(B)$
            \For{$pos \in \{ 1, 2, \ldots, n \}$}
                \State $\mathsf{amount} \gets p_n(pos) \cdot B.\mathsf{TotalDeposited}()$
                \State $B.\mathsf{WithdrawTo}(P[pos].id, \mathsf{amount})$
            \EndFor
        \EndFunction
        \Statex
        \Function{CurrentPhase}{$B$}
            \If{$B.\mathsf{Time}() \le \mathsf{StartTime} + \mathsf{TimePerRound}$}
                \State \Return $\{\mathsf{Commit}\}$
            \ElsIf{$B.\mathsf{Time}() \le \mathsf{StartTime} + 2 \cdot \mathsf{TimePerRound}$}
                \State \Return $\{\mathsf{Open(0)}\}$
            \ElsIf{$B.\mathsf{Time}() \le \mathsf{StartTime} + (2+r) \cdot \mathsf{TimePerRound}$}
                \State \Return $\{\mathsf{Settle(r-1)}, \mathsf{Open(r)}\}$
            \EndIf
            \State \Return $\{\mathsf{Settle}\}$
        \EndFunction
    \end{algorithmic}
    \caption{\sysname Multi-Winner Lottery.}
    \label{alg:multi-winner-lottery}
\end{figure}

\subsection{Perfect Shuffling}
\label{sec:perfect-shuffling}

Instead of just using the two-party lottery,
we can introduce a new similar building block.
This new sub-protocol involves some subset of $m \le n$ out of all $n$ participants.
It is realized as a single round of revealing previously committed values.
The $m$ values are then used to simulate up to $m-1$ sequential two-party lotteries.
Later lotteries in the sequence are only evaluated if the earlier ones did not result in a swap.
So, either all $m-1$ lotteries result in no swap,
or the first swap is applied and all other possible swaps are ignored.

Dependencies between lotteries have to be carefully examined.
Specifically, for fairness we require that aborting an earlier lottery does not increase expected value for any participant in any potential later one.
An adversary controlling multiple participants can thus not increase their expected value by aborting any sub-protocol.

Now onto the specific multiparty lottery construction.
For any round $r \in \{0, 1, \ldots, n-2\}$,
there are $n-r$ participants remaining for consideration.
In every round the remaining participants run the above protocol to determine the last place among themselves.
The participant selected for last place is then removed for the next round.

Consider the unweighted case.
The probabilities are $\frac{1}{n-r}, \frac{1}{n-r-1}, \ldots, \frac{1}{2}$.
In the case where all participants follow the protocol,
this leads to all participants being selected with probability $\frac{1}{n-r}$ each.
Under honest participation, it basically implements a variant of the Fisher-Yates shuffle~\cite{fisher_yates,durstenfeld_fisher_yates}.
That is, it generates all permutations with equal probability.

If any of the $m$ participants does not reveal their value, there are two cases to consider:
Either (i) a swap happened before the participant would have participated in their two-party lottery
or (ii) one of the two-party lotteries needs to be evaluated but one of the two participants aborted.
In the first case, the swap is performed and participants who aborted are then placed in the next places from the back
(reducing the number of remaining rounds accordingly).
In the second case, the two-party lottery is treated as a loss for the participant who aborted,
afterwards any other participants matching the first case are handled.
As previously, in the unlikely case of multiple participants aborting at the same time, ties can be broken arbitrarily.

\begin{theorem}
    If all participants honestly follow the protocol,
    the above construction generates all permutations of the $n$ participants with equal probability.
\end{theorem}

\begin{proof}
    Follows from equivalence to~\cite{fisher_yates,durstenfeld_fisher_yates}.
\end{proof}

\begin{theorem}
    The above construction is fair, in the sense of \cref{def:fairness}.
\end{theorem}

\begin{proof}
    Not being selected for swapping into last place in any round of the protocol eliminates
    at least the single worst option from the possible outcomes (more if other participants abort)
    and other possible outcomes become proportionally more likely.
    Thus, not being selected has non-negative expected utility and being selected has non-positive expected utility.
    Losing or aborting as part of an executed two-party lottery ensures being selected for swapping into last place.
    Aborting as part of an ignored two-party lottery leads to being placed in the next free spot from the back,
    which can be interpreted as aborting or losing in a future round.
    Therefore, aborting in any case has non-positive expected utility.
\end{proof}

In practice, the above construction can be realized in just $n-1$ rounds with $\mathcal{O}(n^2)$ published messages.
This is practically feasible only for small $n$.
However, it solves perfect shuffling and not just correct shuffling.

It could also be adapted to case with weights.
However, it is hard to get the same distribution of permutations as the sequential single-party lotteries.
Since, we are determining the order from last place to first,
in the first round we would need to know the probability of each participant appearing in last place.
Efficiently calculating these seems impossible,
considering their dependence on the exponentially large game tree.
This problem can be circumvented by instead selecting participants proportional to the inverse of their weight.
However, this will change the distribution.

\subsubsection*{Leader Aversion}

This protocol is also a natural fit to solve the leader aversion case,
i.e., where all participants have negative utility for being selected and utility 0 for not being selected.
It can be used to solve single selection under leader aversion in a single round with $\mathcal{O}(n)$ messages
and $k$-leader selection in $k$ rounds with $\mathcal{O}(kn)$ messages.
Specifically, running the first $k$ rounds of the above protocol assigns the last $k$ places in the correct manner.

\section{Privacy-Preserving Lotteries}
\label{sec:unlinkable-lotteries}

In this section we describe how to construct lotteries that (at least partially) preserve the participants' privacy.
Non-interactive zero-knowledge proofs and a cheap optional cooperative sub-protocol
can turn the basic lottery protocols from the previous section into unlinkable ones.
More formally, in addition to the basic lottery properties from \cref{sec:model},
they need to achieve the following property:

\begin{definition}[Unlinkability]
\label{def:unlinkability}
    Consider a multiparty lottery with $n$ participants with equal winning probabilities.
    Given that a targeted user makes $k \le n$ entries to the lottery from the same address.
    An outside observer of the lottery can not have more than $\frac{k}{n}$ confidence that any entry belongs to the targeted address.
    Even an adversary aware of the identities behind $f \le n-k$ of the entries to the lottery,
    can only have $\frac{k}{n-f}$ confidence that any of the remaining entries belongs to the targeted address.
\end{definition}

This property is not easily generalized to arbitrary winning probabilities.
Without heavy usage of zero-knowledge proofs across all rounds of the protocol,
it is likely possible for any observer to see the exact winning probabilities of participants.
All participants choosing the same buy-in is then the only Nash-equilibrium,
which also happens to be Pareto efficient.
Choosing different buy-ins instead would always reduce the size of the anonimity set.

A simple and general way of achieving unlinkability is to break the link immediately between the address purchasing the ticket and the address or addresses participating in the lottery.
As noted before, without any other measures, this requires equal weights for all participants.
This construction is basically equivalent to routing all buy-in payments through a coin mixer.

Compared to the standard variant of \sysname, this protocol has an additional round,
a dedicated ticket purchasing phase before the commitment phase.
In the commitment phase a zero-knowledge proof (ZKP) is published that links a purchased ticket,
i.e., the deposit of the buy-in, to a commitment value, a payout address, and optionally additional participant or relayer addresses.
The linked values are passed along as public inputs to the statement of the ZKP.
This requires the ZKP to be \emph{non-malleable} (see \cref{def:zkp}).
Otherwise, someone could maliciously post a similar proof where they change the payout address to their own.

If a participant intends to use a relayer,
funds to cover their fees should be deposited in addition to the buy-in.
To give further guarantees to the relayer, the ZKP could also include the address of the intended relayer.
Instead of paying the relayer fee to the first party to publish the required value,
it is only paid out if the intended relayer publishes.
This shields the relayer from potential front-running by other network participants,
including validators or block producers of the underlying blockchain.

% This protocol also has an additional phase.
% This phase starts after some player performed the settlement of the lottery.
% Any player can submit a zero-knowledge proof to claim their payout.
% The ZKP proves knowledge of the secret associated with an entry.
% Importantly, if implemented in the regular blockchain setting where transaction contents are publicly visible,
% the ZKP needs to include the payout address and be non-malleable.
% Otherwise, other users of the blockchain could try to front-run the transaction with the same ZKP and their address to claim the reward for themselves.
% The correct amount can then be paid to the provided recipient address.

\begin{figure}[!t]
    \footnotesize
    \begin{algorithmic}
        \State $B$ \Comment{bulletin board}
        \State $P$ \Comment{player IDs (initially empty)}
        \State $N$ \Comment{used nonces (initially empty)}
        \Statex
        \Function{BuyTicket}{$id, s, n$} \Comment{$id$ is the purchase address}
            \State \textbf{assert}\ $\textsf{BuyTicket} \in \textsc{CurrentPhase}(B)$
            \State $y \gets h(s || n)$
            \State $B.\mathsf{DepositFrom}(id, 1)$
            \State $\textsc{Merkle.AddLeaf}(y)$
        \EndFunction
        \Statex
        \Function{Commit}{$id, x(n, id), \pi, V$} \Comment{$id$ is the payout address}
            \State \textbf{assert}\ $\textsf{Commit} \in \textsc{CurrentPhase}(B)$
            \State \textbf{assert}\ $n \not\in N$
            \State \textbf{assert}\ $\textsc{ZKP.Verify}(\pi, x)$
            \State $P \gets P \cup \{id\}$
            \State $N \gets N \cup \{n\}$
            \State $\sysname.\textsc{Commit}(id, V)$ \Comment{excluding the buy-in deposit}
        \EndFunction
        \Statex
        \Function{CurrentPhase}{$B$}
            \If{$B.\textsf{Time}() \le \textsf{StartTime} + \textsf{TimePerRound}$}
                \State \Return $\{\textsf{BuyTicket}\}$
            \ElsIf{$B.\textsf{Time}() \le \textsf{StartTime} + 2 \cdot \textsf{TimePerRound}$}
                \State \Return $\{\textsf{Commit}\}$
            \ElsIf{$B.\textsf{Time}() \le \textsf{StartTime} + 3 \cdot \textsf{TimePerRound}$}
                \State \Return $\{\textsf{Open(0)}\}$
            \ElsIf{$B.\textsf{Time}() \le \textsf{StartTime} + (3+x) \cdot \textsf{TimePerRound}$}
                \State \Return $\{\textsf{Settle(x-1)}, \textsf{Open(x)}\}$
            \EndIf
            \State \Return $\{\textsf{Settle}\}$
        \EndFunction
    \end{algorithmic}
    \caption{\sysnamezkp unlinkable lottery.
        Extends \sysname, omitted functions are the same as shown in \cref{alg:multi-winner-lottery}.}
    \label{alg:unlinkable-lottery}
\end{figure}

\begin{theorem}
    The protocol shown in \cref{alg:unlinkable-lottery} is unlinkable, in the sense of \cref{def:unlinkability}.
\end{theorem}

\begin{proof}
    The zero-knowledge property of the non-interactive zero-knowledge proof ensures that only the truth value of the statement is revealed.
    Namely, it is only revealed that the address activating the ticket is controlled by one of the entities that purchased a ticket.
    So, after the commitment round (and without any external indicators),
    any participants likelihood to belong to any purchasing address is proportional to the number of entries made by that address.
\end{proof}

As already discussed above, this approach does not generalize well to arbitrary winning probabilities.
To this end, we also define a weaker notion of unlinkability,
for which a participant's anonymity depends on the cooperation of the opponents they face during the tournament.
However, we can show how weak unlinkability can be achieved for arbitrary winning probabilities.

\begin{definition}[Weak Unlinkability]
\label{def:weak-unlinkability}
    Consider a multiparty lottery with $n$ participants with weighted winning probabilities.
    Given that a targeted user makes $k \le n$ entries to the lottery with weights $w_1, \ldots, w_k$ from the same address.
    Without loss of generality, call the other weights $w_{k+1}, \ldots, w_n$.
    Assume all participants honestly follow the protocol.
    An outside observer of the lottery can have no more confidence than is implied by the two sets of weights that any payout address in the final ranking is associated with the targeted address.
\end{definition}

It is important to highlight that (strong) unlinkability also depends on honest cooperation of other participants.
In the extreme case, if all other participant reveal their identities,
this inadvertently reveals the identity of the sole remaining participant.
More generally, any participant revealing their identity reduces the anonymity set
and thus increases the confidence in the identities of the remaining participants.

The following protocol achieves weak unlinkability and applies to arbitrary winning probabilities and arbitrary payout functions.
As opposed to the protocol for (strong) unlinkability,
it also only requires a digital signature scheme and a commitment scheme (e.g. hash commitments).

\begin{theorem}
    The protocol shown in \cref{alg:weakly-unlinkable-lottery} is weakly unlinkable, in the sense of \cref{def:weak-unlinkability}.
\end{theorem}

\begin{proof}
    In the case where both parties complete the two-party protocol honestly,
    including the cooperative opening sub-protocol,
    no information about the winner is leaked except for their new commitment value.
    If $v$ is a high-entropy value or if an additional randomizer is added to each iteration of the chained commitment scheme,
    it is not computationally feasible to determine the winner from the new commitment value alone.
    So unless either party reveals their own identity,
    e.g., by publishing their secret values,
    the outcome remains ambiguous for an outsider.
    This is true for any of the two-party lotteries in the tournament.
\end{proof}

Note, that in the given protocol opponents faced higher up in the tournament tree can have a larger impact on the anonymity set.
By revealing their identity they can remove entire sub-trees from the anonymity set.
In particular, if the winner's final opponent reveals their identity,
they can eliminate half of the participants from the anonymity set of potential winners.

\begin{figure}[!t]
    \footnotesize
    \begin{algorithmic}
        \State $B$ \Comment{bulletin board}
        \State $P$ \Comment{player IDs}
        \Statex
        \Function{Open}{$id, x$}
            \If{$\lnot\textsc{CooperativeOpen}(id, i, c_\textsf{new}, v)$}
                \State $\textsc{UnilateralOpen}(id, i, c_\textsf{new}, v)$
            \EndIf
        \EndFunction
        \Statex
        \Function{CooperativeOpen}{$id, r, m, c_\textsf{new}, v)$} \Comment{preserves privacy}
            \State $\textbf{send}\ c_\textsf{new}, v\ \textbf{to}\ \text{opponent}$
            \State $\textbf{wait for}\ c_\textsf{opp}, v_\textsf{opp}$
            \If{timeout while waiting}
                \State \Return \textbf{false}
            \EndIf
            \State $\textsf{winner} \gets v \oplus v_\textsf{opp} \mod 2$
            \If{$id > id_\textsf{opp} \land \textsf{winner} = 1$}
                \State $c \gets c_\textsf{new}$
            \Else
                \State $c \gets c_\textsf{opp}$
            \EndIf
            \State $\sigma \gets \textsf{Sign}(\textsf{winner})$
            \State $\textbf{send}\ \sigma\ \textbf{to}\ \text{opponent}$
            \State $\textbf{wait for}\ \sigma_\textsf{opp}$
            \If{timeout while waiting}
                \State \Return \textbf{false}
            \EndIf
            \If{$id > id_\textsf{opp}$}
                \State $B.\textsf{Publish}(\textsf{Open}(r, m), (c, \sigma, \sigma'))$
            \Else
                \State $B.\textsf{Publish}(\textsf{Open}(r, m), (c, \sigma', \sigma))$
            \EndIf
        \EndFunction
        \Statex
        \Function{UnilateralOpen}{$id, r, c_\textsf{new}, v$} \Comment{reduces privacy}
            \State $\sysname.\textsc{Open}(id, r, c_\textsf{new}, v)$
        \EndFunction
        \Statex
        \Function{SettleMatch}{$r, m$}
            \State $(i, j, p) \gets T[r][m]$
            %\State $B.\textsf{GetValue}(\textsf{Open}(r, g))$
            \If{$(c_\textsf{new}, \sigma_i, \sigma_j) \gets B.\textsf{Read}(\textsf{Open}(r, m))$}
                \If{$\textsf{VerifySig}(\sigma_i, P[i].pk, c_\textsf{new}) \land \textsf{VerifySig}(\sigma_j, P[j].pk, c_\textsf{new})$}
                    \State $P[j].commitment = c_\textsf{new}$
                    \State \Return
                \EndIf
            \EndIf
            \State $\sysname.\textsc{SettleRound}(r, m)$
        \EndFunction
    \end{algorithmic}
    \caption{\sysnamewul weakly unlinkable lottery with optional cooperative opening mechanism.
        Extends \sysname, omitted functions are the same as shown in \cref{alg:multi-winner-lottery}.}
    \label{alg:weakly-unlinkable-lottery}
\end{figure}

The opening mechanism of \sysnamewul can be seen as a drop-in replacement for the two-party lottery.
It can not just be used as part of \sysname.
It could also replace the two-party lottery itself or its use in the simple single-winner protocol.
Most importantly, it can also be used as a drop-in replacement in the case with unequal winning probabilities.
So, applying the \sysnamewul construction to the single-winner multiparty lottery with weighted winning probabilities,
makes it weakly unlinkable.
Also, it does not only turn a lottery into a weakly unlinkable one.
It also reduces the number of transactions required in the good case by almost half.
This comes at the disadvantage of requiring direct communication between participants.

\section{Evaluation}
\label{sec:evaluation}

We implement the \sysname and \sysnamewul protocols in Sui Move~\cite{sui_move} for the Sui~\cite{sui_lutris} blockchain.
These implementations are evaluated based on the transaction fees incurred.
Sui's object model differentiates between owned objects and shared objects.
This is potentially a good fit for \sysname,
since most of the individual two-party lotteries could be evaluated independent of each other.
More importantly, Sui supports a high transaction throughput.
For example, it executes transactions operating on disjoint sets of objects concurrently.
This allows it to support lottery protocols like ours, even with many participants,
better than for example Ethereum would.

\subsection{Implementation}

Sui imposes limitations on total number of bytes any vector (dynamically sized array) can hold.
To allow for potentially very large numbers, e.g. millions, of participants,
we use a BigVector library~\cite{typus_bigvector} developed by Typus Labs.
It uses multiple vectors under the hood,
allowing the resulting data structure to grow past the limits imposed by the underlying system.
The BigVector instances have a parameter (slice size) that needs to be tuned based on the number of participants that should be supported.

As we discussed, committing and opening values can be considered as publishing,
i.e., broadcasting values with minimal requirements for ordering.
The only ordering that needs to happen is with respect to the low-fidelity clock.
This can be made possible by a simple reliable broadcast mechanism.
In the Sui object model this would correspond to using owned objects in most places.
Our current implementation does not yet take advantage of this
and instead stores the entire state of the lottery in a single shared object.

\subsection{Setup}

For evaluation we run the Sui blockchain~\cite{sui_blockchain} in a local testnet.
Interacting with the blockchain and simulating the participants of the lottery
via a Rust implementation based on the Sui Rust SDK~\cite{sui_rust_sdk}.

\subsection{Results}

\begin{table}[!t]
    \centering
    \begin{tabular}[t]{lccc}
        \toprule
        \textbf{Lottery} & \textbf{Txs} & \textbf{Fees [SUI]} & \textbf{Fees [US\$]} \\
        \midrule
        single-winner           & 3,072            & 2.39    & 1.94    \\
        $\text{sequential}^{*}$ & $> 2 \cdot 10^6$ & > 2,300 & > 1,900 \\
        %parallel + PRNG         & 1                &         & 2        \\
        \sysname                & 11,266           & 11.5    & 9.31    \\
        %\sysnamezkp             & 1 & 1 \\
        \sysnamewul             & 6,145            & 5.64    & 4.58    \\
        \bottomrule
    \end{tabular}
    \caption{Practical gas fees and transaction counts of our implementations for the Sui blockchain.
        The given measurements are the totals for running a lottery with 1,024 participants.
        (*) Numbers for the naive sequential lottery are extrapolated from the single-winner case.}
    \label{tab:gas-fees}
\end{table}

\cref{tab:gas-fees} shows example executions of our implementations with 1,024 participants.
These numbers assume the default gas price of 1,000 MIST ($10^{-6}$ SUI),
which is slightly higher than on Sui's mainnet,
and an --- as of this writing --- current price of \$0.812 per SUI.
It can be seen that the absolute gas usage and absolute transaction fees (per user) are relatively low.
Consider that we want to pay at most 1\% in transaction fees.
With the above example of 1,024 participants,
this is achieved with buy-ins of at least \$0.19, \$0.91 and \$0.45
for single-winner, \sysname and \sysnamewul respectively.

\begin{figure}[!t]
    \centering
    \includegraphics[width=\linewidth]{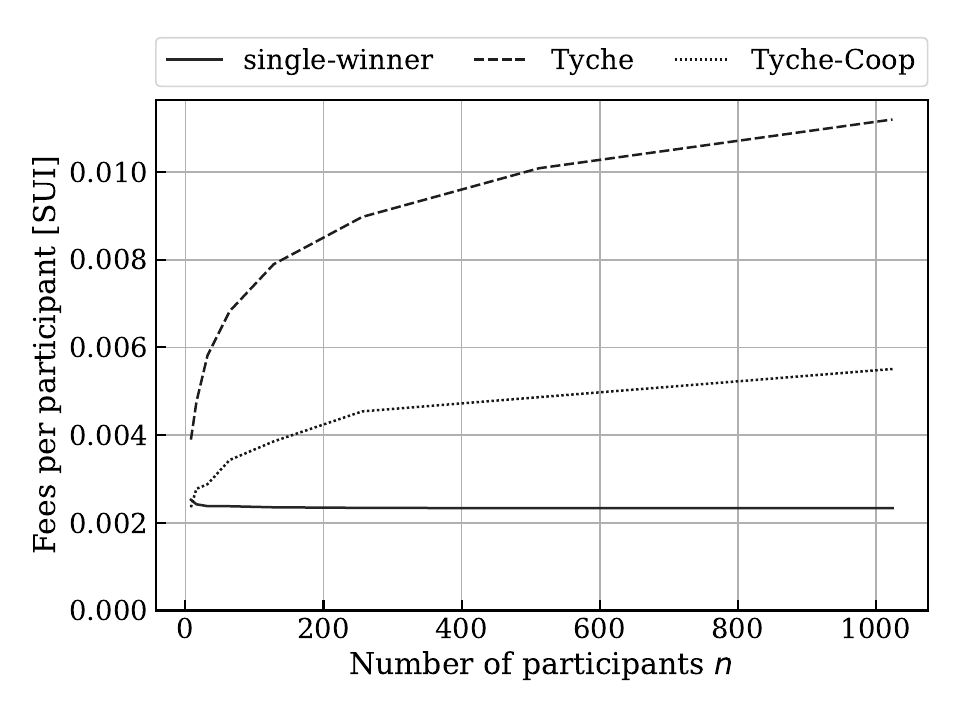}
    \caption{Fees paid by each individual participant as a function of the total number of participants for all of our protocol implementations.}
    \label{fig:gas-fees}
\end{figure}

Further, we see in \cref{fig:gas-fees} how transaction fees paid per participant scale roughly logarithmically in the number of lottery participants for the multi-winner lotteries.
This indicates that the protocols are scalable to many participants,
still with reasonable transaction fees.
It is also interesting to see how much the cooperative opening mechanism of \sysnamewul can reduce transaction fees.
In the experiments above we consider the best case, i.e., all participants participate in the cooperative opening.
This suggests that using \sysnamewul could also be an economic decision, not just one based on privacy concerns.

\section{Conclusion}
\label{sec:conclusion}

In this work we considered how to construct lottery protocols under very weak assumptions.
Our model allows for majority coalitions, does not require any semi-honest third-party (even for liveness),
and does not require participants to post collateral.
We introduced the notion of shuffling networks for thinking about a specific class of lottery protocols.
Based on that, we presented a general framework for constructing multiparty lottery protocols.
Specifically, we gave constructions for the most interesting cases, i.e.,
weighted single-winner lotteries and unweighted multi-winner lotteries.

Many previous protocols relied on a third party to ensure liveness if not even security
or relied on security deposits and punishing misbehaving participants.
Previous work in the same setting did not support arbitrary payout functions (i.e. the multi-winner setting).
\sysname can be seen as a generalization of all lottery protocols in the same setting.
These previous protocols~\cite{two_player_lotteries,zero_collateral_lotteries,purelottery,purelottery_thesis} can be seen as solving certain special cases of \sysname.

\subsubsection*{Future Work}

Further constructions of specific shuffling networks may be studied,
especially constructing a perfect shuffling network of $o(n)$ (ideally $\mathcal{O}(\log n)$) depth.
Any monotonic shuffling network immediately yields corresponding instances of \sysname, \sysnamezkp and \sysnamewul.

Another interesting theoretical problem is whether the general \sysname shuffling network for arbitrary payout functions
can be generalized to support weighted winning probabilities at the same time.
For this an efficient algorithm, or ideally a closed form solution for assigning probabilities to the swappers would be required.

\label{lastpagebeforerefs}

\section*{Research Ethics}

We do not work with any live systems or analyze vulnerabilities in any existing protocols.
This work simply introduces new protocols applying cryptography and game-theory to achieve an
improvement regarding their properties of security, game-theoretic fairness and privacy.
Compared to previous realizations,
our protocols provide the strongest guarantees regarding fairness and transparency.
This gives a direct improvement in these metrics to users when replacing another protocol.

Like any other protocols solving the same problem,
our work may be used to realize digital implementations of games of chance,
where users are able to wager their money.
We see this as a reasonable use of the technology as long as it is handled correctly.
The service should only be offered in a place and manner that is adheres to legal restrictions on users and providers.
In this context, users should always be clearly informed about their odds as well as any fees.
Furthermore, this is by far not the only possible application of the protocols.
They can in general be used as mechanisms for randomly distributing scarce resources
across a population that has a shared utility function over these.

The section on unlinkable lotteries introduces new privacy protocols.
This privacy-enhancing tehcnology protects participants from potential harm or targeting that could arise if their participation or win status were disclosed.
Similar to other privacy preserving payment technologies,
one of the main concerns is that the same feature that protects individual privacy
can be exploited for malicious purposes, such as money laundering or other illicit activities.

\section*{Open Science}

The entire source code resulting from this work will be made available under a permissive open-source license.
This includes the smart contract source code of the lotteries written in the Sui Move language
and the testing and evaluation framework written in Rust using the Sui SDK.
Parts of this implementation are based on code written by Mysten Labs and Typus Labs,
which is also each freely available under the Apache 2.0 license.

% \section*{Acknowledgments}
% The authors would like to thank Kostas Chalkias for motivating this notion of on-chain lotteries
% during an internship that Quentin Kniep did at Mysten Labs.

\bibliographystyle{plain}
\bibliography{references}

\appendix

\end{document}